\documentclass[secnum,leqno]{rims-bessatsu}%%% eqno resets every section 
\usepackage{amssymb, amsmath}%%% option packages
\usepackage[dvipdfm]{graphicx}%%% option packages
\newtheorem{thm}{Theorem}[section]

\theoremstyle{definition}

\newtheorem{exa}[thm]{Example}%%%
\theoremstyle{remark}
\newtheorem*{rem}{Remark}%%% * means no numbering
\title{Geometry and Combinatorics of Crystal Melting}
\TitleHead{Geometry and Combinatorics of Crystal Melting} %optional
\author{Masahito \textsc{Yamazaki}\footnote{Princeton Center for
Theoretical Science, Princeton University, NJ 08544, USA.
\newline
 e-mail: \texttt{masahito@princeton.edu}}
}
\AuthorHead{Masahito Yamazaki}         %optional but recommended
         %optional
\VolumeNo{x}           %editors must define this.
\YearNo{2011}           %editors must define this.
\PagesNo{000--000}      %editors must define this.
\newcommand{\pg}{\overset{+}{\succ}}
\newcommand{\pl}{\overset{+}{\prec}}
\newcommand{\mg}{\overset{-}{\succ}}
\newcommand{\ml}{\overset{-}{\prec}}
\usepackage{youngtab}
\begin{document}
%
% The text goes here.  
% Be sure to use the appropriate "theorem-like" environment as 
% is the following examples.  Never use plain TeX commands for these, as
% they will cause interference with the styles of other papers. 

\maketitle

%\tableofcontents      %optional
\begin{abstract}      %optional
We survey geometrical and especially combinatorial aspects 
%of the recent developements in the theory 
of generalized Donaldson-Thomas invariants 
(also called BPS invariants) for toric
 Calabi-Yau manifolds, emphasizing the role
 of plane partitions and their generalizations in the 
recently proposed crystal melting model. We also comment on equivalence
 with a vicious walker model and the matrix model representation of the partition function \footnote{This is
 a survey article based on the author's talk in the conference
 ``Developments in Quantum Integrable Systems,'' RIMS, Kyoto University,
 June 2010.}.

%This is a submitted to RIMS K\^oky\^uroku Bessatsu.
\end{abstract}

\section{Introduction}

The study of plane partitions, a three-dimensional generalization of
partitions, has a long history of more than a century in mathematics \cite{MacMahon,Stanley}. 
There has recently been a renewed interest in this old topic, both among
mathematicians and physicists alike, due to the pioneering discovery that
topological A-model \cite{GV} on toric Calabi-Yau manifolds \cite{AKMV} can be
described by a statistical mechanical model of plane
partitions \cite{ORV,INOV}. 
In more mathematical language, plane partitions count 
Donaldson-Thomas (DT) invariants \cite{DT,Thomas}, whose partition function
is equivalent \cite{MNOP1,MOOP} to that of the Gromov-Witten invariants under suitable parameter identifications.

There is an interesting twist to this story, which is the topic of more
recent studies in this field. There
we study ``generalized Donaldson-Thomas invariants'' \cite{JS,KS}, which
depend on moduli (mathematically stability conditions, or physically
complexified K\"ahler moduli \footnote{These two notions are precisely speaking not the same, but the differences will not be important in this paper.}). These ``invariants'' are invariant under
 a generic deformation of moduli, 
but can jump when we cross real codimension loci (called walls of
marginal stability, which divide the moduli space into chambers) 
in the moduli space.
This jumping is called the wall crossing
phenomena, and there are general formulas \cite{JS,KS,DenefM} which govern this jumping phenomena.
 Generalized DT invariants are indeed generalizations of the 
original DT invariants, in the sense that the former
coincide with the latter only in a specific chamber (hereafter called the topological string chamber) of the moduli space.

Given the richness of these new invariants, the natural question is
whether there are combinatorial counterparts to this geometric story.
%, generalizing
%the crystal melting model of \cite{ORV}. 
The goal of the present article is to provide an answer to this
question. Generalized DT invariants on toric Calabi-Yau
manifolds are described by a statistical mechanical
model of ``crystal melting'' \cite{Szendroi,MR,OY1}, formulated here 
as an enumeration problem of
 plane partitions and their generalizations. We will also comment on the
 equivalence with a vicious walker model, following \cite{OSY}.
No prior knowledge in this field is assumed,
 and this paper is intended to be self-contained, at least as far
as the combinatorial aspects are concerned.

This article is organized as follows. In section 2 we define our combinatorial partition function as a sum over suitable evolution of partitions. In section 3 we comment on the representation of the partition function as a unitary matrix integral. The derivation of this matrix integral is given in section 4, based on an equivalence with a vicious walker model. Appendix contains an introduction to generalized Donaldson-Thomas invariants, which readers can consult for geometric side of the story.

\section{Definition of the Model}

We begin with the following theorem, which states that the partition function
 for generalized DT invariants $Z_{\rm gDT}$ (see Appendix) on a toric Calabi-Yau manifold can be computed exactly by purely combinatorial methods:
\begin{thm}[Szendr\"oi \cite{Szendroi}, Mozgovoy-Reineke \cite{MR},
 Ooguri-Yamazaki \cite{OY1}] \\
For a toric Calabi-Yau manifold, the partition function for generalized DT invariants can be written as
\begin{align}
Z_{\rm gDT}=Z_{\rm crystal},
\label{DT=crystal}
\end{align}
where the statistical mechanical model in the right hand side can equivalently
formulated as (1) a crystal melting model (a configuration of molten
 atoms), (2) dimer model, or (3) a generalization of
 plane partitions (an evolution of partitions) \footnote{The third formulation is available only when the toric Calabi-Yau manifold has no compact 4-cycles. All the examples in this paper fall into this category.}.  
\end{thm}

The goal of this section is to define $Z_{\rm crystal}$, using the plane partitions and their generalizations. We use the third formulation, developed in
 \cite{YoungBryan,NagaoVO,NY,Sulkowski}. 
See \cite{Szendroi,MR,OY1,Young} for the first and the second description.

\begin{rem}
As we discussed in the introduction, the LHS of the equation
 \eqref{DT=crystal} depends on the value of the moduli. 
Correspondingly, RHS also has moduli dependence, and 
we have different statistical mechanical models for
 different chambers of the moduli space. The schematic relation 
\eqref{DT=crystal} should be interpreted this way. See \eqref{generaleq}.
% this in more detail below.
\end{rem}

Let us begin with standard notations. A partition
$\lambda=(\lambda_i)$ is a non-increasing sequence of integers
$\lambda_1\ge \lambda_2 \ge \ldots \ge 0$ such that $\lambda_n=0$ for
sufficiently large $n$. The length $|\lambda|$ of a partition $\lambda$
is given by $|\lambda| :=\sum_i \lambda_i$.
We define a transpose $\lambda^t$ by 
$$\lambda^t{}_i:=\#\{j| \lambda_i\le j\}.$$
This is simply a graphical transpose of a partition, as will be clear
from the following example.

\begin{exa}
For $\lambda=(4,2,1)={\tiny \yng(4,2,1)}$, $|\lambda|=7$ and
 $\lambda^t=(3,2,1,1)={\tiny \yng(3,2,1,1)}$.
\end{exa}

Given two partitions $\lambda$ and $\mu$, we define $\lambda \pg \mu$ if and
only if 
$$
\lambda_i=\mu_i+1 ~~ \textrm{or} ~~ \mu_i ~~~ \textrm{for all } i.
$$
We also denote $\lambda \mg \mu$ if and only if $\lambda^t \pg \mu^t$,
or equivalently
$$
\lambda_1\ge \mu_1\ge \lambda_2 \ge \mu_2\ge \ldots \,\,.
%\label{C3evolution}
$$

\begin{exa}
Two partitions $\lambda=(4,2,1)={\tiny \yng(4,2,1)}$ and
 $\mu=(3,2,1)={\tiny \yng(3,2,1)}$ both
 satisfy $\lambda\pg \mu$ and $\lambda\mg \mu$.
\end{exa}

Now we define a plane partition (also called a 3d partition) as a sequence of partitions
$\Pi=\{\lambda(n) \}_{n\in \mathbb{Z}}$ such that
\begin{align}
\ldots \pl \lambda(-2) \pl \lambda(-1) \pl \lambda(0) \pg \lambda(1) \pg
\lambda(2) \pg \ldots,
\label{C3evolution}
\end{align}
and $\lambda(n)=\{0 \}$ when $|n|$ sufficiently large.
Define the length $|\Pi|$ of a plane partition $\Pi=\{ \lambda(n)\}$ to
be $|\Pi |=\sum_n |\lambda(n)|$. Note that this is a finite sum by the
assumption above. Our partition function is defined by
\begin{align}
Z_{\rm crystal}(q):=\sum_{\Pi} q^{|\Pi|}. 
\label{C3Z}
\end{align}
As is well-known, MacMahon's formula represents this as an infinite product
\begin{align}
Z_{\rm crystal}(q)=\prod_{k>0}(1-q^k)^{-k},
\end{align}
which is the same as the generalized DT partition function \footnote{See Appendix for summary of these invariants.}
%the topological A-model partition function 
for $\mathbb{C}^3$ \cite{GV}: 
\footnote{In the language of topological strings, $q$ is related to
the topological string coupling constant $g_s$ by 
%\begin{align}
$q=-e^{-g_s}$.
% ~~ Q=e^{-t},
%\label{topparam}
%\end{align}
}
\begin{align}
Z_{\rm crystal}(q)=Z_{\rm gDT}^{\mathbb{C}^3}(q),
\label{C3DT=crystal}
\end{align}
% (K\"ahler
%moduli of the conifold).
For this reason we hereafter denote the LHS of the above equation
\eqref{C3DT=crystal} by $Z_{\rm crystal}^{\mathbb{C}^3}$.

\begin{rem}
In the definition of a plane partition \eqref{C3evolution} we could have used 
$\mg, \ml$ instead of $\pg,\pl$. This has the effect of replacing $\{\lambda(n)\}$ by $\{\lambda(n)^t\}$, and the partition function is the same
 either way.
\end{rem}

\begin{rem}
The choice of the weight in \eqref{C3Z} is the same as in the Schur
 process of \cite{OR}. 
%Our story below for the conifold is a
% generalizati.
\end{rem}

%%%%%%% conifold %%%%%%%%%%%%%%%%%%%%%%%%%%%%
We have seen that plane partitions correspond to the simplest Calabi-Yau
geometry $\mathbb{C}^3$. Our next task is to consider a set of
partitions corresponding to the (resolved) conifold.
We again consider a sequence of partitions $\Pi=\{ \lambda(n)\}$, but
now with a slightly different interlacing conditions, with plus and
minus appearing alternatingly:
\begin{align}
\ldots \ml \lambda(-2) \pl \lambda(-1) \ml \lambda(0) \pg \lambda(1) \mg
\lambda(2) \pg \ldots ~.
\label{conifoldevolution}
\end{align}
For such a $\Pi$, we define $|\Pi|_0:=\sum_{n:\,{\rm even }}|\lambda(n)|$ and 
 similarly $|\Pi|_1:=\sum_{n:\,{\rm odd}} |\lambda(n)|$. The conifold
 partition function is then defined by
\begin{align}
Z_{\rm crystal}^{\rm conifold}(q_0,q_1):=\sum_{\Pi : \,{\rm satisfying }\,
 \eqref{conifoldevolution}} q_0^{|\Pi|_0} q_1^{|\Pi|_1}.
\label{Zconifold}
\end{align}
Then \eqref{DT=crystal} in this example states that 
\begin{align}
  Z_{\rm crystal}^{\rm conifold}(q_0, q_1)=Z_{\rm gDT}^{\rm conifold}(q,Q)
\end{align}
under the parameter identification
\begin{align}
q=q_0 q_1, ~Q=q_1.
\label{conifoldparam}
\end{align}
Infinite product expression for this partition function is known \cite{GV}
from the study of the topological string: \footnote{In the topological strings, 
%\begin{align}
$q=-e^{-g_{\rm s}}, ~~ Q=-e^{-t}$,
%\end{align}
where $g_s$ is the topological string coupling constant, and $t$ is the K\"ahler moduli of the resolved conifold.
}
\begin{align}
Z_{\rm gDT}^{\rm conifold}(q,Q)=M(q)^2 \prod_{k>0}(1+q^k Q)^k \prod_{k>0}
 (1+q^k Q^{-1})^k, 
\end{align}
The corresponding statement for $Z_{\rm crystal}^{\rm conifold}$ was
shown 
combinatorially in \cite{Young}.

\bigskip
The discussion above is a little bit imprecise because we did not
specify the moduli dependence of the generalized DT invariants. 
We therefore consider the following more general partition function
which includes the moduli dependence and applies to more general toric geometries.

First, fix an integer $L$
and a function $\rho: \{1/2,3/2,\ldots, L-1/2\}\to \{\pm 1\}$. We can
 periodically extend $\rho$ to a map $\sigma:\mathbb{Z}_{1/2}\to \{ \pm 1\}$,
where $\mathbb{Z}_{1/2}$ is a set of half-integers. This will determine a toric Calabi-Yau geometry.
 Second, to fix a moduli dependence, we define a bijection $\theta: \mathbb{Z}_{1/2}\to \mathbb{Z}_{1/2}$
 such that  \footnote{This notation including half-integers looks cumbersome, but is useful to see the action of the Weyl group of the affine Kac-Moody algebra \cite{Nagao1, NagaoVO, NY}. The parametrization by $\theta$ actually covers only half of the chambers, but
sufficient for our purposes here. The partition function becomes a
finite product in the remaining half. 
%For comparison with literature, note that the $\theta$ here is $\theta^{-1}$ in \cite{NY}.
}
\begin{align}
\theta(h+L)=\theta(h)+L ~\textrm{ for all }~ h \in \mathbb{Z}_{1/2},
\end{align}
 and 
\begin{align}
\sum_{i=1}^L \theta\left(i-\frac{1}{2}\right)=\sum_{i=1}^L \left(i-\frac{1}{2}\right).
\end{align}
We then define a generalized plane partition of type $(L,\rho,\theta)$ 
 (whose totality we are going to denote by $\mathcal{P}^{(L,\rho, \theta)}$) to be a sequence of partitions $\Pi=\{ \lambda(n)\}$ such that
\begin{align}
\begin{split}
 \lambda(i) \overset{\sigma\circ\theta(i+1/2)}\prec \lambda(i+1) ~~{\rm
 for } ~~ \theta\left(i+\frac{1}{2}\right)<0,\\
 \lambda(i) \overset{\sigma\circ\theta(i+1/2)}\succ \lambda(i+1) ~~{\rm
 for } ~~ \theta\left(i+\frac{1}{2}\right)>0.
\end{split}
\label{generalevolution}
\end{align}
We define $|\Pi|_i:=\sum_{n\equiv i \mod L}|\lambda(n)|$ for
$i=0,\ldots, L-1$. We also define
\[
q^\theta_i:=
\begin{cases}
q_{\theta^{-1}(i-1/2)+1/2}\cdot q_{\theta^{-1}(i-1/2)+3/2}\cdot\cdots\cdot q_{\theta^{-1}(i+1/2)-1/2} &  (\theta^{-1}(i-1/2)<\theta^{-1}(i+1/2)),\\
q_{\theta^{-1}(i-1/2)-1/2}^{-1}\cdot q_{\theta^{-1}(i-1/2)-3/2}^{-1}\cdot\cdots\cdot q_{\theta^{-1}(i+1/2)+1/2}^{-1} & (\theta^{-1}(i-1/2)>\theta^{-1}(i+1/2)),
\end{cases}
\]
where we used $q_i$ for $i\in \mathbb{Z}$ by extending $q_0,\ldots, q_{L-1}$ periodically,
$$
q_{i+L}=q_i.
$$
Note that $q_i^{\theta=\textrm{id}}=q_i$.
The partition function is defined by
\begin{align}
Z_{\rm crystal}^{(L,\rho,\theta)}(q_0, q_1,\ldots, q_{L-1}):=\sum_{\Pi\in \mathcal{P}^{(L,\rho,\theta)}} (q^{\theta}_0)^{|\Pi|_0}
 (q^{\theta}_1)^{|\Pi|_1} \ldots (q^{\theta}_{L-1})^{|\Pi|_{L-1}}
\label{Zgeneral}
\end{align}

\begin{exa}
Let us take $L=1$, $\rho=+1$ and then the only choice for $\theta$ is
 $\theta={\rm id}$, and \eqref{generalevolution} reduces to
 \eqref{C3evolution}.
 This means $Z_{\rm crystal}^{(L=1,\rho=+1,\theta=\textrm{id})}=Z_{\rm crystal}^{\rm \mathbb{C}^3}$.
% \eqref{generalrevolution} reduces to \eqref{C3evolution}.
\end{exa}

\begin{exa}
Take $L=2$, $\rho(1/2)=+1$ and $\rho(3/2)=-1$. $\theta$ can in general written
 as 
$$\theta=\theta_n: ~ \frac{1}{2} \mapsto \frac{1}{2}-n,~ \frac{3}{2}\mapsto \frac{3}{2}+n.
$$
When $\theta=\theta_0$, \eqref{generalevolution} is the same as
 \eqref{conifoldevolution}, and the partition function \eqref{Zgeneral} coincides with \eqref{Zconifold}. % under the identification \eqref{conifoldparam}.
The case of $n\ne 0$ corresponds to generalized DT invariants in other chambers. The corresponding partition function is given by \cite{NN,JM,AOVY} (again under the identification 
\eqref{conifoldparam})
\begin{align}
Z_{\rm crystal}^{\rm conifold}(q_0,q_1;\theta_n)=M(q)^2 \prod_{k>0} (1+q^k
 Q)^k \prod_{k>n}(1+q^k Q^{-1})^k.
\label{ZBPS-n-conifold}
\end{align}
In particular, in the limit $n\to \infty$, this coincides with the
 commutative DT partition function for the conifold \cite{GV}. In this
 limit our statistical mechanical model reduces to the gluing of
 two crystal corners (topological vertices) as in \cite{ORV,AKMV}.
%different chambers in
% the conifold example.
%The general $n$ corresponds to a
% different chamber. We denote the corresponding partition function by
% $Z_{\rm crystal}(q,Q;n)$. The goal of this next section is to give a
% matrix model representation for this partition function.
\end{exa}

The general story goes as follows. We can construct from $(L,\rho)$ a toric Calabi-Yau manifold
$X^{(L,\rho)}$, which is one of the so-called generalized conifolds. 
This is a toric Calabi-Yau manifold without compact 4-cycles \footnote{See \cite{Aganagic} for discussion of Calabi-Yau geometries with a compact 4-cycle.}, and has
a connected
string of $L-1$
$\mathbb{P}^1$'s. Each $\mathbb{P}^1$ is either a
$\mathcal{O}(-2,0)$-curve or a $\mathcal{O}(-1,-1)$-curve, depending on
$\sigma(i-1/2)=\sigma(i+1/2)$ or $\sigma(i-1/2)=-\sigma(i+1/2)$. \footnote{This means that the overall sign change of $\rho$ does not change the geometry. In \eqref{generalevolution}, this has the effect of replacing $\pg,\pl$ by $\mg,\ml$. As discussed previously in the case of $\mathbb{C}^3$, this does not change the partition function, but will change the matrix model representation of the partition function presented in the next section.} We can
then consider the partition function of generalized DT invariants on
$X^{(L,\rho)}$. 

The remaining task is to specify the moduli dependence, which in this case is given by an element of the Weyl group of the affine Lie algebra $\hat{A}_{L-1}$ \cite{Nagao1}. The corresponding partition function is denoted by $Z_{\rm gDT}^{X^{\rho}}(q,Q;\theta)$.
Now the following theorem states that this
partition function is the same as the crystal partition function of type
$(L,\rho,\theta)$:
\begin{thm}[Nagao \cite{Nagao1}]
\begin{align}
 Z^{\sigma}_{\rm crystal}(q_0, \ldots, q_{L-1};\theta)=Z_{\rm gDT}^{X^{\rho}}(q,Q;\theta),
\label{generaleq}
\end{align}
\end{thm}
where the parameter identifications are given by \footnote{The signs are
determined from $\rho$. See \cite{MR} and section 3.5 of \cite{Yamazaki}.}
\begin{align*}
q=\pm q_0\ldots q_{L-1}, ~~Q_i=\pm q_i ~~(i=1,\ldots, L-1).
%\label{generalparameteridentification}
\end{align*}

\begin{rem}
We can consider further generalizations, by changing the boundary conditions at infinity. This generalized model counts ``open generalized Donaldson-Thomas invariants''. See \cite{NagaoVO,NY,AY,Sulkowski2}.
\end{rem}

In the following we concentrate on the case of $\mathbb{C}^3$ and the resolved conifold.

%%%%%%%%%%%%%%%%%%%%%%%%%%%%%%%%%%%%%%%%%%%%%%%%%%%%%%%%%%%

\section{Matrix Model}

In the following sections we show that the crystal melting partition function defined in the previous section can be written as a unitary matrix integral:
\begin{thm}
\begin{align}
Z_{\rm crystal}^{\mathbb{C}^3}(q)=\lim_{N\to \infty}\int_{U(N)} dU \det \Theta(U|q),
\label{C3matrix}
\end{align}
where $dU$ is the Haar measure of the unitary group and
\begin{align} 
\Theta(u| q) = \prod_{k=0}^\infty (1 + u q^k)(1+u^{-1}q^{k+1}).
\end{align}
%is a product of two quantum dilogarithm functions.
\label{thm3.1}
\end{thm}

\begin{thm}[Ooguri-Su{\l}kowski-Yamazaki \cite{OSY}]
\begin{align}
Z_{\rm crystal}^{\rm conifold}(q,Q;\theta_n)=C_n Z_{\rm matrix}^{\rm conifold}(q,Q;n),
\label{eq4.3}
\end{align}
where
\begin{align}
C_n=\prod_{k=1}^{n} \frac{1}{(1-q^k)^k} \prod_{k=n+1}^{\infty}
\left( \frac{1-Q^{-1}q^k}{1-q^k}\right)^n,
\label{Cnconifold}
\end{align}
and 
\begin{align}
Z_{\rm matrix}^{\rm conifold}(q,Q;n)=\lim_{N\to \infty}\int_{U(N)} dU \det \left(
 \frac{\Theta(U|q)}{\Theta(QU|q)}
\prod_{k=1}^n (1+Q^{-1}U^{-1} q^k)\right),
\label{conifoldmatrix}
\end{align}
where the measure $dU$ and the function $\Theta(u|q)$ are the same as in Theorem 3.1.
\label{thm3.2}
\end{thm}
%The integer $N$ is introduced as a regulator, and we take a limit $N\to \infty$.
We shall give derivations of these results in the next section, but
before going there some comments are in order.

\begin{rem}
Theorem \ref{thm3.1} and Theorem \ref{thm3.2} for $n=0$ are not new, although $n\ne 0$ case has not previously appeared in the literature as far as the author is aware of. Theorem 3.1 seems to be well-known in the literature, and 
can be considered as a reduction of a multi-matrix model in \cite{EynardTASEP}.
 For Theorem 3.2 with $n=0$, see the paper \cite{BaikRains}, which also proves similar identities for other groups. A vicious walker model similar to the one for the conifold presented in the next section was also discussed in \cite{Hikami}. In their terminology $\pg, \pl$ ($\mg, \ml$) are called
 ``normal'' (``super'') time evolution, and the model is analyzed by
%RHS of Theorem \ref{thm3.2}
%and the the expression \eqref{conifoldf} 
identifies involving semi-standard Young tableaux and hook
 Schur functions (also called supersymmetric Schur functions).
 They also analyze the scaling limit of the model. See also \cite{Szabo}.
\end{rem}

\begin{rem}
The prefactor $C_n$ simplifies in the limit $n=0$ and $n\to\infty$; $C_0=1$ and $C_{\infty}=M(q)$. In particular in these cases $C_n$ is independent of $Q$.
\end{rem}

\begin{rem}
$Z_{\rm matrix}$, being a partition function of a unitary matrix
 integral, is a reduction of a $\tau$-function of a two-dimensional
 integrable Toda chain \cite{Takasaki} (see e.g. \cite{Kharchev}). Similar integrable structures appeared in topological strings context in \cite{ADKMV}. See also \cite{GMN2,CNVTBA} for the appearance of thermodynamic Bethe Ansatz equations in the study generalized DT invariants.

%Toda chain
%$$
%n^2 \tau_{n+1}\tau_{n-1}=\tau_n \tau_n^{''}-(\tau_n^{'})^2.
%$$
%general chamber: a reduction of 2-dimensional Toda hierarchy.
\end{rem}

%\begin{rem}

Finally, let us discuss the thermodynamic limit of our model, using the
matrix integral given above. 
The thermodynamic limit is the limit
$g_s\to 0$, where the string coupling constant $g_s$ is related to the
parameter $q$ by $q=e^{-g_s}$ \footnote{This is the parameter counting
the size $|\Pi|:=\sum_{i=1}^{L-1} |\Pi|_i$ of a generalized plane partition $\Pi$.}.
%
%We treat both $\bC^3$ and the Since the chemical potential for $|\Pi|$
%is $q$ (see \eqref{generalparameteridentification}), 
%identification of $q=e^{-g_s}$, where $g_s$ is the string coupling constant. 
%
For small $g_s$, the modular transformation of $\Theta$
with respect to $g_s$ gives
\begin{align*}
 \Theta(e^{i\phi}| e^{-g_s}) = 
e^{-\frac{\phi^2}{2g_s}}\cdot \left(1  +  O(e^{-\frac{1}{g_s}})\right).   
\end{align*}
If we ignore non-perturbative terms in $g_s$, this means that 
the matrix model $Z_{\rm matrix}^{\mathbb{C}^3}$ reduces to
 a Hermitian Gaussian matrix model with unitary measure. 
This result was originally
 derived from the Chern-Simons theory on the conifold \cite{mm-lens,CSmatrix} (see also \cite{Okuda} for crystal melting description).

The spectral curve for this Gaussian matrix model 
is given by the equation \cite{CSmatrix,Marino}
\begin{align}
  e^x + e^y + e^{x-y - T} + 1 = 0, 
\end{align}
where $T=Ng_s$ is the 't Hooft coupling. This is the mirror of the resolved conifold. 
In the limit of $T\rightarrow \infty$ (which we should take since $N\to \infty$), the curve reduces to
\begin{align}
  e^x + e^y + 1 = 0, 
\end{align}
which is the mirror of $\mathbb{C}^3$. 

Next we discuss the spectral curve for the conifold matrix model \eqref{conifoldmatrix} \footnote{This is the spectral curve for the matrix model. Our statistical model can equivalently be written as a dimer model, which has its own version of the spectral curve. See \cite{KOS} and \cite{OY2}.}.
As before, we take the limit $g_s\to 0, N\to \infty$ with $T:=N g_s$
fixed, but now we also take $n\to \infty$, with $\tau:=n g_s$ fixed.
The spectral curve is given by \cite{OSY}
\begin{align}
e^{x+y}+e^x +e^y +Q_1 e^{2x}+ Q_2 e^{2y}+Q_3=0,
\label{conifoldcurve}
\end{align}
where
\begin{align}
\begin{split}
& Q_1 = \epsilon^2 \cdot 
\frac{1+\mu Q}{(1+\mu \epsilon^2)(1+Q\epsilon^2)}, \\
& Q_2 = \mu \cdot
 \frac{1+Q\epsilon^2}{(1+\mu Q)(1+\mu\epsilon^2 )}, \\
& Q_3 = Q \cdot \frac{1+\mu \epsilon^2}{(1+\epsilon^2 Q)(1+\mu Q)},
\label{Q123}
\end{split}
\end{align}
and $\mu= Q^{-1} q^n, \epsilon^2=e^{-T}$.
This is the mirror \cite{HoriVafa} of the so-called closed vertex geometry, 
whose web diagram is 
shown in Figure \ref{closedvertex}.

\begin{figure}[htbp]
\centering{\includegraphics[scale=0.25]{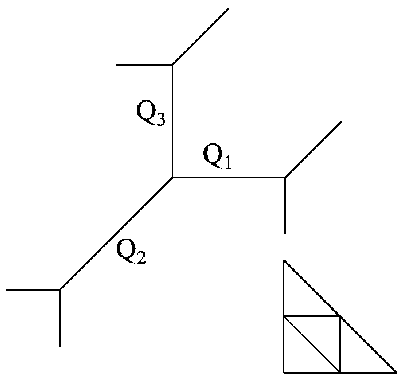}}
\caption{The web diagram for the close vertex geometry. Three 
$\mathbb{P}$'s with size $Q_1, Q_2, Q_3$ appear symmetrically.}
\label{closedvertex}
\end{figure}

There are two interesting observations on this result.
First, \eqref{Q123} coincides with the mirror map for this geometry.
Second, the curve (\ref{conifoldcurve}) 
is symmetric under exchanges of $Q$, $\mu=Q^{-1}q^n$ and $\epsilon^2 
= e^{-T}$. Namely, (1) the original K\"ahler moduli $Q$ 
of the resolved conifold, (2) the chamber parameter $n$ and (3) the 't Hooft
parameter $T$ appear symmetrically in the spectral curve. 
This is an interesting result, which suggests a possible connection between continuum limit of the wall
 crossing formulas and the BCOV
 holomorphic anomaly equation \cite{BCOV}.

In the matrix model we are interested in the limit of $N\to
\infty$, which means $T \rightarrow \infty$ or equivalently $\epsilon \rightarrow 0$.
With appropriate shifts of $x$ and $y$, the equation 
(\ref{conifoldcurve}) in this limit becomes
\begin{align}
  \mu\ e^{2y} + e^{x+y} + e^x + (1 + Q\mu)\ e^y + Q = 0.
\label{SPPcurve}
\end{align}
The is the mirror of the so-called Suspended Pinched Point (SPP) geometry,
with $Q$ and $\mu$ being exponentials of 
flat coordinates representing sizes of its two $\mathbb{P}^1$'s, 
which encode two copies of the initial $\mathcal{O}(-1,-1)\to\mathbb{P}^1$ geometry, see figure \ref{SPP}. 
Not only does the spectral curve agree
 with the mirror curve of the SPP geometry in the 
limit of $g_s \rightarrow 0$, but in fact
the matrix integral reproduces the full topological string 
partition function all orders in $g_s$ expansion.
Indeed, it is known that the SPP topological string partition function \footnote{In our context, a topological string partition function is a generalized DT partition function in the topological string chamber, where generalized DT invariants coincide with the original DT invariants of \cite{DT,Thomas}. This chamber is the analogue of $\theta_{n=\infty}$ in the conifold},
with K{\"a}hler parameters $Q$ and $\mu$, is equal to
\begin{align}
 Z^{\rm SPP}_{{\rm top}}(q, Q, \mu) 
= \prod_{k=1}^\infty \frac{(1-Qq^k)^k(1-\mu q^k)^k}
{(1-q^k)^{3k/2}(1 - \mu Q q^k)^k}.     \label{ZtopSPP}
\end{align}
On the other hand, from the explicit structure of the BPS generating
function and formulas (\ref{ZBPS-n-conifold}), (\ref{eq4.3}), (\ref{Cnconifold}) and (\ref{conifoldmatrix}),
we find that the value of the matrix integral, in the $N \rightarrow \infty$ limit, 
is related to the above topological string partition function as
\begin{align}   
  Z_{{\rm matrix}}(q, Q; n)  =  Z^{\rm SPP}_{{\rm top}}(q, Q, \mu = Q^{-1}q^n)
  \cdot \prod_{k=1}^\infty (1-q^k)^{k/2}.
 %& = & \prod_{k=1}^\infty (1-q^k)^k \cdot \prod_{j=1}^{\infty} \frac{(1-Qq^j)^j(1-\mu q^j)^j}{(1-\mu Q q^j)^j }
 \label{uptomacmahon}
\end{align}
This result is consistent with the philosophy of the remodeling
conjecture \cite{BKMP}, which states that a set of invariants (symplectic invariants) \cite{EO} constructed recursively from the spectral curve coincide with topological string partition function on the same geometry. Since symplectic invariants are defined by rewriting loop equations of matrix models purely in the language of spectral curves, the fact that the topological string partition function can be written as  a matrix model would prove the remodeling conjecture.
Indeed, this type of logic was used in \cite{EynardTASEP,EynardTopological} to prove the
 remodeling conjecture for toric Calabi-Yau manifolds.
%, although the results in these papers are in the topological string chamber and 
It would be interesting to know whether similar recursion relations exist 
in other chambers.

\begin{figure}[htbp]
\centering{\includegraphics[trim=0 0 0 0,scale=0.25]{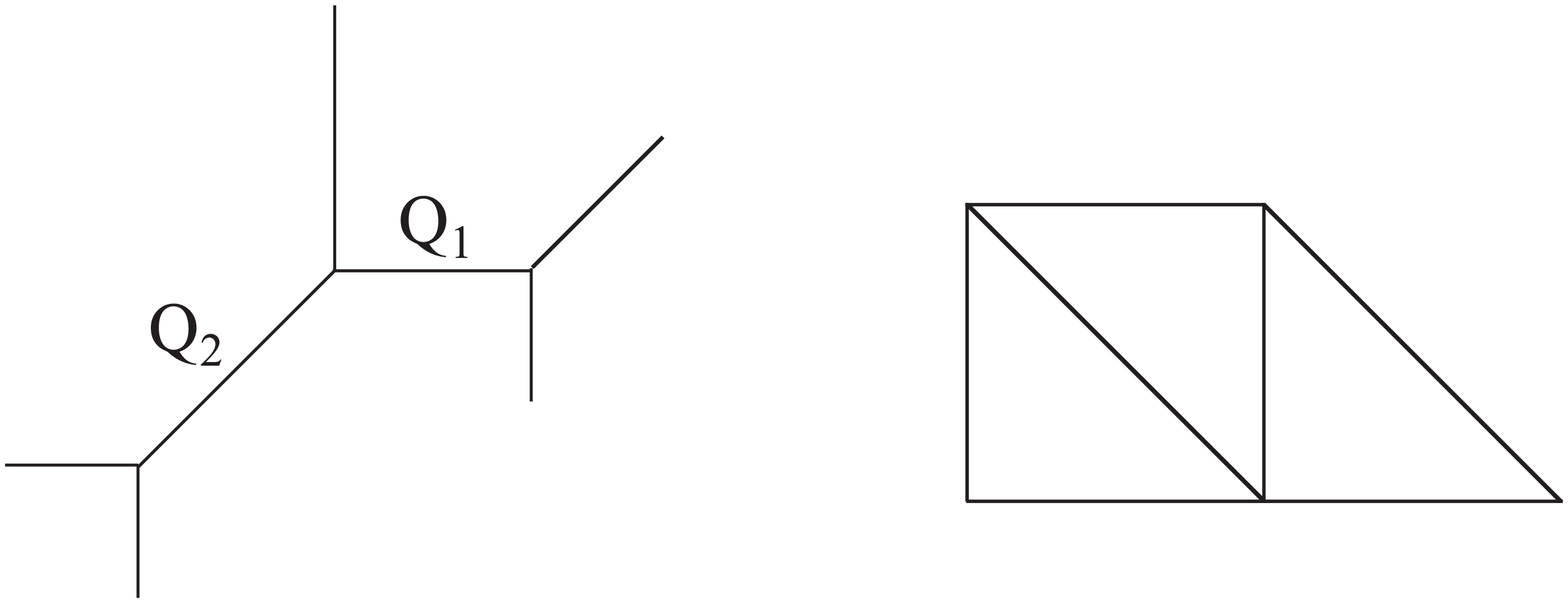}}
\caption{The web diagram for the SPP geometry. This geometry has two $\mathbb{P}^1$'s with size $Q_1$ and $Q_2$.}
\label{SPP}
\end{figure}

%\begin{rem}
%Connection with open topological strings
%\end{rem}

%%%%%%%%%%%%%%%%%%%%%%%%%%%%%%%%%%%%%%%%%%%%%%%%%%%%%%%%
\section{Derivation of the Matrix Models}

There are at least two derivations of the above-mentioned matrix models, one using the vertex operator formalism for free fermions \cite{JimboMiwa}
(see \cite{YoungBryan,NagaoVO, Sulkowski} for discussion our context) 
and another using the equivalence
with a vicious walker model (non-intersecting paths). %\cite{Fischer}. 
Both are presented in \cite{OSY}.
Here we comment on the 
latter method in the case $\theta={\rm id}$. 
The derivation here is a slightly simplified version of the 
derivation in \cite{OSY}. See also \cite{EynardTASEP,Szabo,EynardTopological}, 
which constructs similar matrix models in a particular chamber. In particular \cite{EynardTopological} treats arbitrary toric Calabi-Yau geometries.

Let us fix a sufficiently large number $N$. This is the same $N$ for the
 size of the unitary matrix in the previous section, and in the end we take the limit $N\to \infty$. 
Define 
\begin{align}
 h_k(t):=\lambda_{N-k+1}(t)+k-1, ~~ (k=1,\ldots, N).
\label{hkdef}
\end{align}
%Since each $\lambda_k(t)$ is a partition, we have
%%\begin{itemize}
%%\item 
%\begin{align}
% h_k(t)< h_{k+1}(t) ~~{\rm  for ~~ all}~~ t,
%\label{C3nonintersecting}
%\end{align}
%and since $\lambda(t)=\{0\}$ for $|t|$ sufficiently large,
%%\item 
%\begin{align}
%h_k(t)=k-1 ~~{\rm for}~~ |t|~~ {\rm large}.
%\label{C3boundary}
%\end{align}
%The interlacing condition \eqref{C3evolution} is translated into
%%\item 
%\begin{equation*}
%h_k(t+1)-h_k(t)= 0 ~{\rm or}\, -1,
%\end{equation*}
%for $t \geq 0$ and 
%\begin{equation*}
%h_k(t+1)-h_k(t)= 0 ~{\rm or}~ 1,
%\label{Cpm}
%\end{equation*}
%for $t < 0$. 
%\end{itemize}
%
Since $\lambda(t)$ is a partition, we have
\begin{align}
h_k(t)< h_{k+1}(t),
\label{C3nonintersecting}
\end{align}
for all $t$. We also have the boundary condition,
\begin{align}
h_k(t)=k-1 ~\mathrm{when} ~|t| ~\textrm{large} .
\label{C3boundary}
\end{align}
Moreover, \eqref{C3evolution} means we have, for each step $t$,
\begin{equation}
h_k(t+1)-h_k(t)= 0 ~{\rm or}\, -1,
\label{pm1}
\end{equation}
for $t \geq 0$ and 
\begin{equation}
h_k(t+1)-h_k(t)= 0 ~{\rm or}~ 1,
\label{pm2}
\end{equation}
for $t < 0$.

Suppose that we fix a large positive (negative) integer $t_{\max}$ ($t_{\min}$). We are going to send these numbers of infinity. If we plot the value of 
$\{h_k(t)\}_{t=t_{\rm min}}^{t_{\rm max}}$ for each $k$, we have a set of $N$ 
paths. Due to the conditions \eqref{pm1}, \eqref{pm2} the paths move on the graph
shown in Figure \ref{C3path}, and \eqref{C3boundary} means we have a
fixed boundary condition. Finally, \eqref{C3nonintersecting} means that
 $N$ paths are non-intersecting. Summing up, we have a statistical mechanical model of non-intersecting paths (also called a vicious walker model \cite{Fisher} \footnote{We can also regards this model as a time evolution of $N$ particles in one dimension. In this language the model is an exclusion process, a variant of the ASEP \cite{ASEP}.}), whose partition function is given by:
\begin{align}
Z\simeq \sum_{ \{h_i(t)\}: \textrm{ non-intersecting paths on the graph}}\, \prod_t q^{\sum_i h_i(t)},
\end{align}
where $\simeq$ shows that we neglected an overall multiplicative constant.

\begin{figure}[htbp]
\centering{\includegraphics[scale=0.25]{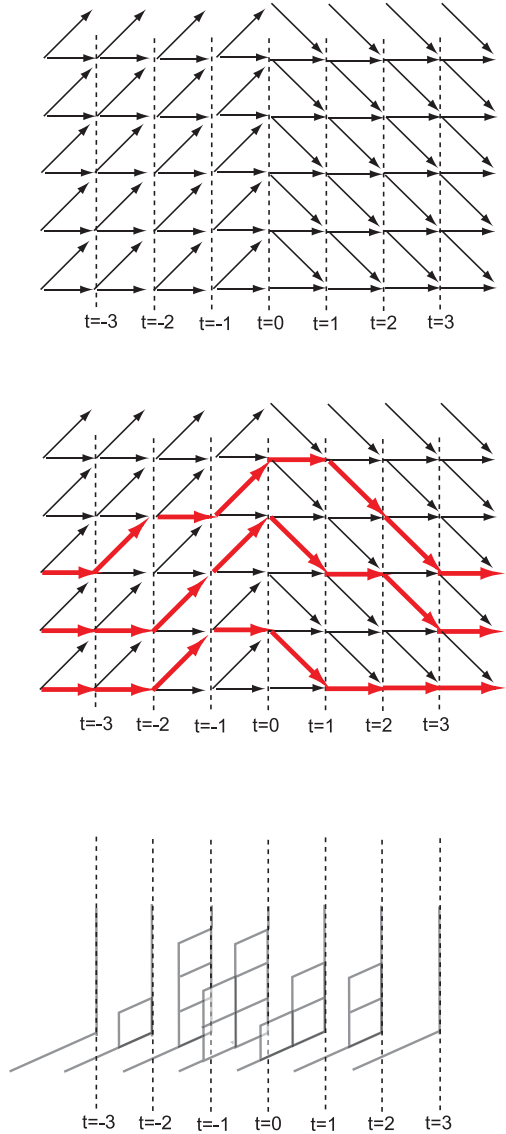}}
\caption{Top: an oriented graph for $\mathbb{C}^3$. Middle:
an example of 3 non-intersecting paths shown as bold (red) arrows. 
The location of the $k$-th path at time $t$ gives $h_k(t)$. Bottom: The corresponding evolution of Young diagrams.}
\label{C3path}
\end{figure}

At first sight it seems difficult 
in practice to implement the non-intersecting conditions for paths.
The following theorem states that we can write the sum over
non-intersecting paths as a determinant of a matrix, whose 
element is defined by a single path:

\begin{thm}[Lindstr{\"o}m \cite{Lindstrom}, Gessel-Viennot \cite{GesselV}; Karlin-McGregor \cite{KarlinMcgregor}]
Suppose we are given an oriented graph without oriented loops. Suppose
 moreover that each edge $e$ comes with a
 weight $w(e)$.
 For a path $p$ on the graph, we define $w(p)$ to be the product of the
 weighs for all the edges on
 the path:
$w(p):=\prod_{e\in p}w(e)$. We define $F$ by summing over all non-intersecting
paths $\{ p_i\}$ (each $p_i$ starts from  $a_i$ and ends at $b_i$): 
\begin{align}
F(\{a_i\},\{b_i\})=\sum_{ \{ p_i:\, a_i\to b_i\}: \textrm{ non-intersecting}} \prod_i w(p_i),
\end{align}
and an $N\times N$ matrix $G(a_i,b_j)$ by
\begin{align}
G(a_i,b_j)= \sum_{p \textrm{: a path from } a_i \textrm{ to } b_j} w(p).
%G(a_i, b_j):=\sum_{path: a_i\to b_j} w(p)
\end{align}
Then
\begin{align}
F(\{a_i\},\{b_i\})=\det_{i,j}\,(G(a_i,b_j)).
%F(\{a_i\},\{b_i\})=\det_{i,j} (G (a_i,b_j))
\end{align}
\label{LGV}
\end{thm}

\begin{proof}
When we expand the determinant $\det_{i,j}\, (G(a_i,b_j))$, we 
have contributions from non-intersecting as well as intersecting paths.
However, contributions from the latter cancel out because
they always come in pairs with an opposite sign (Figure \ref{LGVproof}) \footnote{When more than two paths intersect at a single point, we need to pick two of them according to a fixed ordering and apply the same argument.}.
\end{proof}

\begin{figure}[htbp]
\centering{\includegraphics[scale=0.4]{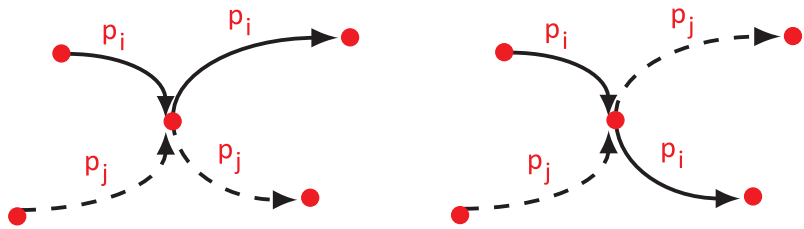}}
\caption{When we expand $\det G$, intersecting paths 
always come in pairs with an opposite sign. The reason is that
 we can exchange the label for paths after 
the intersection, without changing the paths themselves.}
\label{LGVproof}
\end{figure}

\begin{exa}
Consider an oriented graph with weights $w_1,\ldots, w_6$ as shown in Figure \ref{LGVeg}.
It is easy to see that
\begin{align}
F(\{a_1,a_2\},\{b_1,b_2\})=(w_1 w_4)(w_3 w_5),
\end{align}
and
\begin{align}
G(a_i,b_j)=
\left(
\begin{array}{cc}
 w_1 w_4 ~~& w_1 w_6 \\
 w_2 w_4 ~~& w_2 w_6+ w_3 w_5
\end{array}
\right).
\end{align}
We indeed have $\det G=w_1 w_4 (w_2 w_6+w_3 w_5)-(w_2 w_4)(w_1 w_6)=F$.
\begin{figure}[htbp]
\centering{\includegraphics[scale=0.5]{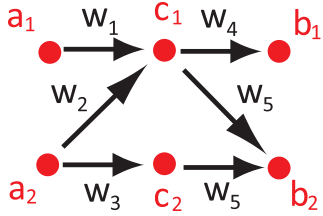}}
\caption{An example of an oriented graph.}
\label{LGVeg}
\end{figure}
\end{exa}

\begin{rem}
The fact that paths are non-intersecting is a manifestation of 
free fermions, and the determinant is interpreted as 
a Slater determinant.
\end{rem}

Therefore we have
\begin{align}
 Z_{\rm crystal}^{\mathbb{C}^3}(q)=\det_{i,j}\, (G_{i,j}(q)),
\end{align}
where $G_{i,j}(q)$ is defined as a weighted sum over all possible paths 
which start at height $i$ at $t=t_{\rm min}$ and end at height $j$ at
$t=t_{\rm max}$. 
As we can see from Figure \ref{C3path},
$G(a_i, b_j )$ depend only on the difference $i-j$,
and we thus have
\begin{align}
Z_{\rm crystal}^{\mathbb{C}^3}(q)= \det_{i,j}\, (G_{i-j}(q)).
\end{align}

%\begin{figure}[htbp]
%i-j
%\caption{i-j}
%\label{i-j}
%\end{figure}

It turns out to be easier
 to write a generating function for $G_n$
than to write each  $G_n$ separately:
\begin{align}
f(z):=\sum_n G_n z^n=\prod_n (1+z q^n ) \prod_n (1+z^{-1} q^{n+1}).
\label{C3f}
\end{align}
\begin{proof}
To see this, note that a term in the expansion of the product is in
one-to-one correspondence with a path. For example, for $t<0$ we 
take either $1$ or $z q^{t}$
from the product, and the choice corresponds to the 
two possibilities in \eqref{pm2}.
The change of the horizontal
coordinates is measured by $z$, and taking the coefficient in front 
of $z^n$ means summing paths with height change $n$.
The product in \eqref{C3f} is over all non-negative integers $n$ when we send $t_{\rm min}\to -\infty, t_{\rm max}\to \infty$.
\end{proof}

Finally, we have the following theorem:
\begin{thm}[Heine \cite{Heine}, Szeg\"{o} \cite{Szego}]
Suppose that $f(z)=\sum_n G_n z^n$. We then have 
\begin{equation}
\int_{U(N)} dU \det f(U)=\det_{1\le i,j \le N}\, G_{i-j}
\end{equation}
\label{Heine}
\end{thm}

\begin{rem}
The RHS of the equation is often called a Toeplitz determinant of $f$.
\end{rem}

\begin{proof}
Diagonalize the unitary matrix $U$ to be $(e^{\sqrt{-1} \phi_1}, \ldots,
 ,e^{\sqrt{-1} \phi_N})$ . Then the integral $\int dU$ reduces to
 $\int \prod d\phi_i \det_{i,j} (e^{\sqrt{-1} i \phi_j }) \det_{i,j}
 (e^{-\sqrt{-1} i\phi_j})$, while the integrand becomes a product $\prod_i
 f(e^{\sqrt{-1}\phi_i})$. 
After expanding the two determinants using the definition of the
 determinant,
we can easily carry out the integral, and the result follows.
\end{proof}
This theorem, together with the form of $f(z)$ in \eqref{C3f}, completes 
the derivation of the matrix model for $\mathbb{C}^3$.

\bigskip
%------------ conifold --------------------------
The analysis for the conifold is 
essentially the same, so let us summarize the result
 briefly.
By defining $h_k(n)$ again as in \eqref{hkdef}, we again have
\eqref{C3nonintersecting} and \eqref{C3boundary}, 
except that \eqref{pm1}, \eqref{pm2} are going to be replaced by 
\begin{enumerate}
\item When $t$ is odd, 
\begin{align}
h_k(t+1)-h_k(t)=
\begin{cases}
0, ~ 1 & (t < 0), \\
0,  ~ -1 & (t \ge 0).
\end{cases}
\end{align}

\item When $t$ is even, 
\begin{align}
\ldots \leq h_{k-1}(t+1) <h_k(t)\leq h_k(t+1) < h_{k+1}(t) \leq \ldots .
\label{eq.hevolve1}
\end{align}
for $t< 0$ and 
\begin{align}
\ldots \leq h_{k-1}(t) <h_k(t+1)\leq h_k(t) < h_{k+1}(t+1) \leq \ldots .
\label{eq.hevolve2}
\end{align}
for $t\ge 0$.
\end{enumerate}
These conditions mean $\{ h_k(n)\}$ move on the graph shown in Figure
\ref{conifoldpath}. Note that the structure of the graph is 
different depending on whether $t$ is even or odd.

\begin{figure}[htbp]
\centering{\includegraphics[scale=0.25]{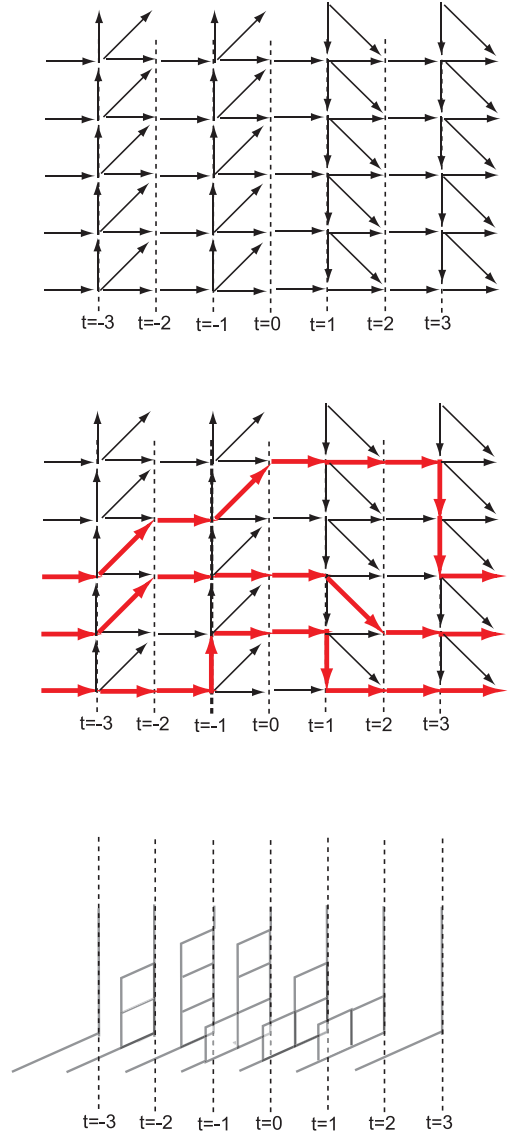}}
\caption{Top:
an oriented graph for the conifold.
Middle: an example of 3 non-intersecting paths on the graph shown
in red. Bottom: the corresponding evolution of Young diagrams.}
\label{fig.conifoldorientedgraph}
\label{conifoldpath}
\end{figure}

Again by using Theorems \ref{LGV} and \ref{Heine}, we have
\begin{align}
Z=\int dU \det f(U),
\end{align}
where
\begin{align}
f(U)= \prod_n \frac{1+(q_0 q_1)^n z}{1-(q_0 q_1)^n q_1 z} \prod_n
 \frac{1+(q_0 q_1)^{n+1} z^{-1}}{1-(q_0 q_1)^{n+1} q_1 z^{-1}}.
\label{conifoldf}
\end{align}
This is nothing but the expression \eqref{conifoldmatrix} for $n=0$.

\section*{Acknowledgments}

The author would like to express his gratitude to the organizers and the
participants of the workshop for creating a stimulating atmosphere. The material here was
presented also at a seminar at the IPMU, University of Tokyo, and the
 feedback from the audience there is greatly acknowledged. The author
 would also like to thank K.~Nagao, H.~Ooguri and P.~Su{\l}kowski for
 related collaboration and discussion. This work was supported in part
 by JSPS postdoctoral fellowships for young scientists, MEXT, Japan, and by Princeton Center for Theoretical Science.

\appendix

\section{Generalized Donaldson-Thomas Invariants}

%We begin with a geometrical setup. 
In this appendix we briefly summarize the ingredients of the generalized Donaldson-Thomas invariants. Our discussion in this section is
far from rigorous and at best schematic, since the main focus of this paper is more on combinatorial aspects presented in the main text. 
%Readers in hurry can proceed directly to next section and come back here when necessary.
%not interested in geometrical aspects of the story can safely skip this section.

For the definition of generalized DT invariants, we need the following:
\begin{itemize}
\item $X$: a Calabi-Yau 3-fold.

\item a ``charge lattice'': %$H_{\textrm{even}}(X;\mathbb{Z})$. 
\[
 H_{\textrm{even}}(X;\mathbb{Z})=H_0(X;\mathbb{Z})\oplus H_2(X;\mathbb{Z}) \oplus H_4(X;\mathbb{Z}) \oplus H_6(X;\mathbb{Z}).
\]

\item complexified K\"ahler moduli of $X$:
\[
 t_i=B_i+\sqrt{-1}\, k_i, ~~~ (i=1,\ldots, \textrm{dim}\,H_2(X;\mathbb{Z})),
\]
where the real part $B_i$ (imaginary part $k_i$) denotes the B-field
      flux through (the volume of) the $i$-th 2-cycle.

\item A central charge function $Z_{\gamma}(t)$, which depend on
      $t:=\{t_i\}$ and linearly on $\gamma\in
      H_{\textrm{even}}(X;\mathbb{Z})$ \footnote{This is part of the data
      for the stability conditions \cite{Bridgeland}.}.
\end{itemize}
With these data we can define ``generalized DT invariants'' 
\[
 \Omega(\gamma; t) \in \mathbb{Q}.
\]
For concreteness, in this paper we restrict ourselves to the following situations:
\begin{itemize}
\item $X$: a toric \footnote{Toric here means local toric, i.e. it is a
      canonical bundle over a complex 2-dimensional toric variety.}
      Calabi-Yau 3-fold without compact 4-cycles. For example, $X$ can
      be $\mathbb{C}^3$ or the (resolved) conifold. 

\item In the charge lattice $H_{\textrm{even}}(X;\mathbb{Z})$ we only consider the
      following set of charges \footnote{
%The meaning of cohomology is
      %subtle since $X$ is non-compact. 
Here $H_0$ and $H_2$ correspond
      to compact 0- and 2-cycles, respectively, and $H_6$ to $X$ itself,
      which is noncompact.}: 
\begin{align*}
\begin{split}
H_{\textrm{even}}(X;\mathbb{Z}) & = H_0(X;\mathbb{Z})\oplus H_2(X;\mathbb{Z})\oplus H_4 (X;\mathbb{Z})\oplus H_6 (X;\mathbb{Z})\\
%\rotatebox{90} & {$\in$} ~~  \rotatebox{90}{$\in$}\\
\gamma& =(~~n,~~ ~~~~~~~~ \beta=\{\beta_i\},~~~~~~~~~~ 0, ~~~~~~~~~~1~~).
\end{split}
\end{align*}
\end{itemize}
We then have a set of integer invariants
\begin{align*}
 \Omega\left(\gamma\!=\! (n,\beta,0,1);\,t\right) \in \mathbb{Z}.
\end{align*}
Instead of studying these invariants separately, it is useful to define
their generating function:
\begin{align}
 Z_{\rm gDT}(q,Q;t)=\sum_{n,\beta} \Omega\left(\gamma\!=\!(n,\beta,0,1);\,t\right)\, q^n Q^{\beta},
\end{align}
where $Q:=\{Q_i\}$ denotes a set of parameters and $Q^{\beta}:=\prod_i
Q_i ^{\beta_i}$. This is the partition function for generalized DT invariants studied in the main text.
% For other descriptions, see the references above.
%, although often call all three
%equivalent models by simply ``crytal melting'' \footnote{In the literature the word ``crystal melting'' often refers specifically to the
%counting of plane partitions as in \cite{ORV}. Our counting problem of
%crystal melting, however, is in general different from theirs (due to
% the wall crossing phenomena) and should
%be considered as a generalization of the old story.}. 

%%%%%%%%%%%%%%%%%%%%%%%%%%%%%%%%%%%%%%%%%%%%%%%%%%

\end{document}